\documentclass{article}

\usepackage[maxbibnames=99]{biblatex}
\usepackage{amsmath,amssymb,amsfonts,amsthm}
\usepackage{algorithmic}
\usepackage{graphicx}
\usepackage{textcomp}
\usepackage[dvipsnames]{xcolor}
\usepackage{url}
\usepackage[ruled, linesnumbered, vlined, noend]{algorithm2e}

\newcommand{\bI}{\mathbf{I}} 
\newcommand{\bP}{\mathbf{P}} 
\newcommand{\bS}{\mathbf{S}} 
\newcommand{\bD}{\mathbf{D}}

\newcommand{\bH}{\mathbf{H}}

\newcommand{\bA}{\mathbf{A}} 
\newcommand{\adj}{\mathsf{Adj}} 
\newcommand{\bB}{\mathbf{B}}

\newcommand{\bx}{\mathbf{x}}
\newcommand{\by}{\mathbf{y}}
\newcommand{\adjlpha}{\boldsymbol{\alpha}}

\newcommand{\F}{\mathbb{F}}
\newcommand{\Z}{\mathbb{Z}}

\usepackage{mathtools}

\newcommand{\leftquote}{``}
\newcommand{\rightquote}{''}

\newcommand{\acode}{\mathcal{A}}
\newcommand{\bcode}{\mathcal{B}}

\newcommand{\eqdef}{\stackrel{\text{def}}{=}}

\makeindex

\DeclareMathOperator{\GL}{GL}

\theoremstyle{plain}

\newtheorem{theorem}{Theorem}

\newtheorem{property}{Property}
\newtheorem{proposition}[theorem]{Proposition}

\newtheorem{definition}[theorem]{Definition}

\newtheorem{remark}[theorem]{Remark}

\numberwithin{equation}{section}

\usepackage{hyperref}

\allowdisplaybreaks

\usepackage{authblk}

\author[1]{Michele Battagliola}
\author[2]{Anna-Lena Horlemann}
\author[3]{Abhinaba Mazumder}
\author[4]{Rocco Mora}
\author[5]{Paolo Santini}
\author[6]{Michael Schaller}
\author[7]{Violetta Weger}

\affil[1,5]{Marche Polytechnic University, Italy}
\affil[1,2]{University of St.Gallen, Switzerland}
\affil[4]{University of Montpellier, France}
\affil[3,6]{University of Zurich, Switzerland}
\affil[7]{Technical University of Munich, Germany}

{
    \makeatletter
    \renewcommand\AB@affilsepx{: \protect\Affilfont}
    \makeatother

    \affil[ ]{Email}

    \makeatletter
    \renewcommand\AB@affilsepx{, \protect\Affilfont}
    \makeatother
    \affil[1]{battagliola.michele@proton.me}
    \affil[2]{anna-lena.horlemann@unisg.ch}
    \affil[3]{abhinaba.mazumder@math.uzh.ch}
    \affil[4]{rocco.mora@lirmm.fr}
    \affil[5]{p.santini@staff.univpm.it}
    \affil[6]{michael.schaller@math.uzh.ch}
    \affil[7]{violetta.weger@tum.de}
}

\title{The Power of Power Codes: New Classes of Easy Instances for the Linear Equivalence Problem}

\begin{document}

\maketitle

\begin{abstract}
Given two linear codes, the Linear Equivalence Problem (LEP) asks to find (if it exists) a linear isometry between them; as a special case, we have the Permutation Equivalence Problem (PEP), in which isometries must be permutations.
LEP and PEP have recently gained renewed interest as the security foundations for several post-quantum schemes, including LESS.
A recent paper has introduced the use of the Schur product to solve PEP, identifying many new easy-to-solve instances. In this paper, we extend this result to LEP. 
In particular, we generalize the approach and rely on the more general notion of power codes.
Combining it with Frobenius automorphisms and Hermitian hulls, we identify many classes of easy LEP instances.
To the best of our knowledge, this is the first work exploiting algebraic weaknesses for LEP.
Finally, we show an improved reduction to PEP whenever the coefficients of the monomial matrix are in a subgroup of the multiplicative group of the finite field.
\end{abstract}

\section{Introduction}

The Code Equivalence Problem (CEP) asks whether there exists a Hamming-metric isometry mapping one linear code to another. 
Whenever the isometry is required to be a permutation, we speak of the Permutation Equivalence Problem (PEP); when instead we consider the more general case of monomials (permutations that can also scale coordinates), then we speak of the Linear Equivalence Problem (LEP).
In recent years, PEP and LEP have gained a renewed interest because of their use as security assumptions in post-quantum cryptographic schemes such as \cite{biasse2020less, barenghi2022advanced, battagliola2025vole, hanzlik2025tanuki, albrecht2025hollow, baldi2025speck}.

Interestingly, attacks on CEP comprise several mathematical tools and techniques, ranging from the use of low weight codewords \cite{beullens2020not, budroni2025two}, weight enumeration of the hull \cite{SSA}, graph theory \cite{magali}, canonical forms \cite{chou2025linear, nowakowski2025improved}, algebraic solvers \cite{saeed2017algebraic} and the Schur product \cite{shur}.

Interestingly, these attacks reveal a rather different situation regarding the hardness of PEP and LEP.
Indeed, several weak instances for PEP are known: whenever the hull is too small, then PEP can be solved in average polynomial time  
\cite{SSA, magali, shur}.
Instead, when it comes to LEP, the only currently known weak instances are those defined over $q$-ary finite fields with $q < 5$ \cite{sendrier2013hardness}.
For $q\geq 5$, no weak instance is known: a random LEP instance with non-trivial parameters is currently deemed hard to solve\footnote{As we discuss later on, up to polynomial factors, LEP can be solved in worst case time $O(q^{k})$, with $k$ being the code dimension.
When $k$ is too small, e.g., it is a constant or $O\big(\log(n)\big)$, we have efficient solvers.}.

\subsubsection*{Our contribution}

In this paper, we identify several previously unknown weak instances. At the basis of our techniques lies the use of \emph{power codes}, exploiting the property that the powers of equivalent codes remain equivalent.
This approach, combined with the concepts of the Frobenius automorphism and the (Hermitian) hull, allows us to derive new easy LEP instances based on the length and dimension of the code.

Remarkably, prior to this paper, power codes in the context of CEP have only been used to attack PEP \cite{shur}. 
Here, we show how power codes can be used to attack (certain instances of) LEP.
Looking ahead, the classes of weak instances we identify, differ significantly from the ones that were already known, which required only $q\leq 4$.
Indeed, our attacks have a strong algebraic flavor and show weak instances for any finite field.
These instances have $k$ sublinear in $n$ and therefore do not affect the security of cryptographic schemes such as LESS.
However, in our case the dimension is allowed to grow as a power of $n$ instead of just polylogarithmically, in which case SSA gives a quasi-polynomial time algorithm.
For instance, the first class we identify requires $k<\sqrt[r]{r!\cdot n}\cdot \big(1+o(1)\big)$, for $r$ an integer such that $q = 2r+1$. 
We further improve the attack for non-prime fields.

Furthermore, for a subgroup $U \subseteq \F_q^\star$, we introduce an intermediate problem LEP($U$) between LEP and PEP and the concept of the \emph{partial-closure}, which gives a reduction from this novel problem to PEP.
Together with power codes this can be used to improve the complexity of the reduction from LEP to PEP compared to the classical reduction.

\section{Preliminaries}\label{sec:preliminaries}

In this section we introduce some preliminary facts. 

\subsection{Notation}

For $q$ a prime power, we denote by $\mathbb{F}_q$ a finite field with $q$ elements and by $\mathbb{F}_q^\star$ its multiplicative group.
We denote matrices/vectors by bold upper/lower case letters.
The identity matrix of size $k$ is denoted by  $\bI_k$. By $\GL_n(\mathbb{F}_q)$ we denote the group of $n \times n$ invertible matrices over $\mathbb{F}_q$.
In addition, we denote by $S_n$ the symmetric group on $n$ elements, and by $\mathrm{diag}(\mathbf d)$ the diagonal matrix whose diagonal is $\mathbf d$.
We denote by 
$\langle \bx, \by \rangle$ the  standard inner product, {i.e.}, 
$\langle \bx, \by \rangle \eqdef \sum_{i=1}^n x_iy_i,$
and by $\bx * \by \eqdef(x_1 y_1,\dots,x_n y_n)$ the component-wise product or \emph{Schur product} of $\bx$ and $\by$. Finally, we denote by $\otimes$   the Kronecker product.

\subsection{Coding Theory Basics}

We now provide some basics on coding theory, see also \cite{Huffman_Pless_2003}.
An $[n,k]_q$ \emph{linear code} $\acode$  is a $k$-dimensional linear subspace of $\F_q^n$. We call $\bA \in \F_q^{k \times n}$ a generator matrix of $\acode$ if its rows form a basis of $\acode$.    The \emph{dual code} of $\acode$ is an $[n,n-k]_q$ linear code  defined as $
\acode^\perp \eqdef \{ \bx \in \F_q^n : \langle \bx, \by \rangle = 0 \text{ for all } \by \in \acode \}$.
A generator matrix $\bH \in \F_q^{(n-k) \times n}$ of  $\acode^\perp$ is called \emph{parity-check matrix} of $\acode$. In particular, we have that $\bA\bH^\top = 0$.

\begin{definition}[Hull]
The \emph{hull} of a linear code $\acode$ is defined as
$\mathcal{H}(\acode) \eqdef \acode \cap \acode^\perp$.

\end{definition}
The dimension of the hull of a random code has been studied in \cite{SendrierHull}. In particular, a random code has a trivial hull  with probability $ \geq 1 - 1/q - 1/q^2,$ see \cite[Lemma 1]{albrecht2025hollow}.

\begin{definition}[Schur Product and Power Code]
    Let $\acode, \bcode$ be linear codes.
    We define the \emph{Schur product} of $\acode$ and $\bcode$ denoted by $\acode \star \bcode$ as
    \[
        \acode \star \bcode \eqdef \langle \{ \bx * \by \mid \bx \in \acode, \by \in \bcode\}\rangle. 
    \]
    Let $\ell$ be a positive integer.
    The $\ell$-th \emph{power code} of $\acode$ is defined as
    $$\acode^{(\ell)} \eqdef \acode^{(l-1)} \star \acode $$
    with $\acode^{(1)} = \acode$.
    When $\ell=2$ we call $\acode^{(2)}$ the \emph{square code}.
\end{definition}

If $\acode$ is an $[n,k]_q$ linear code, then the dimension of $\acode^{(\ell)}$ is upper bounded by $\min\left\{ \binom{k+\ell-1}{\ell}, n\right\}$.

\begin{definition}[Hermitian Inner Product] Let $q=p^{2m}$, for some prime $p$ and positive integer $m.$
For $\mathbf  x,  \mathbf y \in \mathbb{F}_q^n$ let us denote by $\langle  \mathbf x, \mathbf  y \rangle_H$ the \emph{Hermitian} inner product, {i.e.}, 
$\langle \mathbf x, \mathbf  y \rangle_H \eqdef \sum_{i=1}^n x_iy_i^{p^m}.$
\end{definition}

Analogously to the previous definitions we can also define the Hermitian dual, parity-check matrix and hull ($\mathcal{H}_H(\acode)$), by replacing the standard with the Hermitian inner product.

\begin{definition}[Closure]
   Let $\acode$ be an $[n,k]_q$ linear code with generator matrix $\bA\in \F_q^{k\times n}$, let $\alpha \in \mathbb{F}_q$ be a primitive element, and denote $ \adjlpha=(1, \alpha, \ldots, \alpha^{q-2}) \in \mathbb{F}_q^{q-1}$. The \emph{closure} $\widetilde{\acode}$ of $\acode$ is defined as the code with generator matrix $\adjlpha \otimes \bA$.
The closure is an $[n(q-1),k]_q$ code.
\end{definition}

\subsection{The Code Equivalence Problem}

\begin{definition}[Linear  Equivalence Problem (LEP)]\label{def:CE}
Given as input two generator matrices $\bA,\bB\in\mathbb F_q^{k\times n}$  of two linear codes $\acode, \bcode\subseteq \mathbb F_q^n$, the Linear  Equivalence Problem asks to find, if there exists, an invertible matrix $\bS \in \mathrm{GL}_k(\mathbb{F}_q)$, a diagonal matrix $\bD = \mathrm{diag}(\mathbf d)$ with $\mathbf d \in (\F_q^\star)^n$, and a permutation matrix $\bP \in S_n$ such that
\[
\bB = \bS \bA \bD \bP.
\]
If $\bD$ is required to be the identity matrix we call this problem the \emph{Permutation Equivalence Problem (PEP)}.
\end{definition}

Whenever the dimension $h$ of the hull of the input code is small enough, then PEP can be solved on average in polynomial time relying on SSA \cite{SSA} or in worst case quasi-polynomial time by the reduction to the Graph Isomorphism Problem (GIP) \cite{magali}.
The complexity of both algorithms grows as $\widetilde O(n^h)$.
Hence, whenever $h = O\big(\log(n)\big)$, we obtain a quasi-polynomial time solver for PEP.

Moreover, \cite{hull} shows how to transform any LEP instance into a PEP instance, using the closure.
\begin{theorem}[Reducing LEP to PEP]\cite[Theorem 1]{hull}
\label{thm:reduction_LEP_PEP}
    Let $\acode, \bcode$ be two $[n,k]_q$ linear codes. Then $\acode$ is linearly equivalent to $ \bcode$ if and only if 
    $\widetilde{\acode}$ is permutation equivalent to $\widetilde{\bcode}$.
\end{theorem}

The above transformation is, however, not very useful from a cryptanalytic standpoint, due to the following fact.
\begin{proposition}
Let $\acode$ be an $[n,k]_q$ linear code, and let $\widetilde{\acode}$ be its closure.
Then, if $q\geq 5$, $\widetilde{\acode}\subset \widetilde{\acode}^\bot$, hence the hull of $\widetilde{\acode}$ has dimension $k$. 
\end{proposition}
If $k$ is too small (say, $k = O\big(\log(n)\big)$), then the resulting PEP instance can be solved in quasi-polynomial time by SSA.
However, whenever the dimension is larger,  we obtain a hard instance for PEP.

\section{Power Codes}
\label{sec:power_codes}

 The main idea of this section is as follows. If $\acode$ and $\bcode$ are linearly equivalent, then their $(q-1)$-th power codes are permutation equivalent. For sufficiently low rate, this observation can be leveraged to recover the underlying permutation, and hence the associated monomial transformation. We refine this approach to extend its applicability to a wider range of parameters. 

Let $\acode$, $\bcode$ be two $[n,k]_q$ linear codes with generator matrices $\bA,\bB$ respectively.
In \cite{magali}, the authors define a function $\textsf{Adj}$ that, given a generator matrix $\bA$ for a linear code with trivial hull, computes the adjacency matrix of an undirected, weighted graph.
We generalize this function: for two $k\times n$ matrices $\bA, \bB\in\mathbb F_q^{k\times n}$ such that $\bA  \bB^\top \in \GL_k(\mathbb{F}_q)$, we define
$$\adj(\acode,\bcode) \eqdef \adj(\bA,\bB) \eqdef \bB^\top (\bA\bB^\top)^{-1}\bA.$$
This definition does not depend on the choice of the generator matrices and for $\bA = \bB$, our function corresponds to the one from \cite{magali}.

Given $\acode,\bcode\subseteq\9F_q^n$ with generator matrices $\6A, \6B$, respectively, we show a deterministic construction of four $[n,k']_q$ codes $\mathcal{A}_1,\mathcal{A}_2, \mathcal{B}_1,\mathcal{B}_2$,   with generator matrices $\6A_1,\6A_2,\6B_1, \6B_2$ with the following property:
\begin{property} \label{property_equivalence}
If $\6B=\6S\6A\6D\6P$ as in Definition \ref{def:CE}, then $\6B_1=\6S_1 \6A_1 \6D^i \6P, \;\6B_2=\6S_2 \6A_2 \6D^j \6P$, with $(q-1)\mid (i+j), 1\le k'<n$, for some $\6S_1,\6S_2\in\mathrm{GL}_{k'}(\mathbb{F}_q)$.
\end{property}
These matrices $\6A_1,\6A_2,\6B_1, \6B_2$ will be obtained using certain Schur product codes. 

Let $\mathsf{Count}$ be the function that, given a vector over $\F_q$, outputs the multiset of entries of that vector, i.e., it counts how many times each element of $\F_q$ appears.
Algorithm~\ref{alg:diag} calls $\mathsf{Count}$ on the vector given by the diagonal entries of the adjacency matrices to decide whether $\acode$ and $\bcode$ are equivalent. We show that this test has no false negatives (but it might have false positives).

\begin{algorithm}[t]
\KwData{finite field size $q$, code length $n$}
\KwIn{$[n,k]_q$ linear codes $\acode,\bcode$}
\KwOut{\leftquote TRUE\rightquote\ or \leftquote FALSE\rightquote}
\vspace{2mm}
\SetAlgoNoLine
Generate $[n,k']_q$ linear codes $\mathcal{A}_1,\mathcal{A}_2,\mathcal{B}_1,\mathcal{B}_2$  from $\acode, \bcode$ satisfying Property \ref{property_equivalence}\;
$\6X=\textsf{Adj}(\mathcal A_1,\mathcal{A}_2)$;\text{ } $\6Y=\textsf{Adj}(\mathcal B_1,\mathcal{B}_2)$\;
\tcc{Compare multiset of the diagonals}
\Return $\mathsf{Count}(x_{1,1},\ldots,x_{n,n}) == \mathsf{Count}(y_{1,1},\ldots,y_{n,n})$;
\caption{Simple solver for LEP\label{alg:diag}
}
\end{algorithm}
\begin{proposition} \label{prop: generalapproach_LEP}
    Let  $\acode, \bcode$ be two $[n,k]_q$ linear codes with generator matrices $\6A,\6B$ respectively. Let $\mathcal{A}_1,\mathcal{A}_2,\mathcal{B}_1,\mathcal{B}_2$ be $[n,k']_q$ codes with generator matrices $\6A_1,\6A_2, \6B_1,\6B_2$ respectively, such that,
    \begin{itemize}
        \item $\6B=\6S\6A\6D\6P \Rightarrow \6B_1=\6S_1 \6A_1 \6D^i \6P, \;\6B_2=\6S_2 \6A_2 \6D^j \6P$, with $(q-1)c=(i+j)$, for some $c \in \Z_{>0},$ and for some $\6S_1,\6S_2\in\mathrm{GL}_{k'}(\mathbb{F}_q)$, i.e., satisfying Property \ref{property_equivalence}.
        \item $\mathcal{A}_1\cap\mathcal{A}_2^\bot=\mathcal{B}_1\cap\mathcal{B}_2^\bot=\{\mathbf{0}\}.$
    \end{itemize}
    If Algorithm~\ref{alg:diag} on input $\mathcal{A}$ and $\mathcal{B}$ outputs ``FALSE'', then $\mathcal{A}$ and $\mathcal{B}$ are not linearly equivalent.
\end{proposition}
\begin{proof}
    We prove the proposition by contraposition, assuming that $\mathcal{A},\mathcal{B}$ are linearly equivalent. 
    By hypothesis, we have $\6B_1=\6S_1 \6A_1 \6D^i \6P, \;\6B_2=\6S_2 \6A_2 \6D^j \6P$, with $(q-1)c= (i+j)$ for some positive integer $c$ and $\6S_1,\6S_2\in\mathrm{GL}_{k'}(\mathbb F_q)$. 
    Recall that $(\6D^j)^\top=\6D^j$ and that $\6D^i\6D^j=\6D^{i+j}=\6D^{c(q-1)}=\6I_n$, since $a^{q-1}=1$ for any $a\in\9F_q^\star$. Hence, $\textsf{Adj}(\bB_1,\bB_2)$ corresponds to
\begin{align*}&\6B_2^\top(\6B_1 \6B_2^\top)^{-1}\6B_1\\
    =&(\6S_2 \6A_2 \6D^j \6P)^\top ((\6S_1 \6A_1 \6D^i \6P)(\6S_2 \6A_2 \6D^j\6P)^\top)^{-1}(\6S_1 \6A_1 \6D^i \6P)\\
    =&\6P^\top \6D^j \textsf{Adj}(\bA_1,\bA_2) \6D^i \6P.
\end{align*}
Given an $n\times n$ matrix $\6M=(m_{u,v})_{u,v}$, we have $\6D^j\6M\6D^i=(d_{u,u}^j d_{v,v}^i m_{u,v})_{u,v}$. 
In particular, the elements on the diagonal remain unchanged by the transformation:
\[
d_{u,u}^j d_{u,u}^i m_{u,u}=d_{u,u}^{c(q-1)} m_{u,u}=m_{u,u}.
\]
Applying $\6P^\top$ on the left and $\6P$ on the right implies an identical swap of column and row indices. Therefore, the elements on the diagonal are simply permuted among them. It follows that $\mathsf{Count}(x_{1,1},\ldots,x_{n,n}) = \mathsf{Count}(y_{1,1},\ldots,y_{n,n})$ in 
Algorithm~\ref{alg:diag}, hence the output is ``TRUE''.
In other words, whenever Algorithm~\ref{alg:diag} outputs ``FALSE'', $\mathcal{A}$ and $\mathcal{B}$ can not be linearly equivalent.
\end{proof}
However, false positives might happen, when $\mathsf{Count}$ returns the same value, i.e., the diagonals of $\textsf{Adj}(\bA_1,\bA_2)$ and $\textsf{Adj}(\bB_1,\bB_2)$ are permutations of each other, although the two codes $\mathcal{A},\mathcal{B}$ are not linearly equivalent.
Estimating the probability of having false positives is not trivial.
It also potentially depends on the specific construction of the codes $\mathcal{A}_1,\mathcal{A}_2,\mathcal{B}_1,\mathcal{B}_2$, which we have not discussed yet.
However, we have experimentally verified that such a probability can be approximated by the probability that two \textit{random} vectors in $\9F_q^n$ are a permutation of each other, see Table~\ref{tab: LEP_prob}.
For fixed field size $q$, and $n\to \infty$ the latter grows as $q^{q/2}(4\pi n)^{(1-q)/2}$ (cf. \cite[Theorem 4]{richmond2009counting}) so
\begin{align} \label{eq: prob_countdiag}
\9P(\mathsf{Count}(x_{1,1},...,x_{n,n})\neq \mathsf{Count}(y_{1,1},...,y_{n,n}))= \nonumber\\
=1-q^{q/2}(4\pi n)^{(1-q)/2},
\end{align}
i.e., it is polynomially small in the length. Therefore, Algorithm~\ref{alg:diag} succeeds with overwhelming probability.

Thus, it remains to explicitly determine the codes $\mathcal{A}_1,\mathcal{A}_2,\mathcal{B}_1,\mathcal{B}_2$, given the pair $(\acode,\bcode)$. Such a construction depends on the field size $q$, which also heavily impacts the range of applicability of our approach. The crucial condition for our method to work is that the pairs of codes $(\mathcal{A}_1,\mathcal{B}_1)$ and $(\mathcal{A}_2,\mathcal{B}_2)$ should not be related by any other isometries other than the one connecting $\acode,\bcode$.
Hence, a necessary condition clearly becomes that $\mathcal{A}_1,\mathcal{A}_2,\mathcal{B}_1,\mathcal{B}_2$ must not be the full space.
Although our constructions do not generally guarantee to avoid false positives our experiments show that the corresponding probability becomes negligible as $k'$ from Property \ref{property_equivalence} becomes sufficiently small compared to $n$, see Table \ref{tab: LEP_prob}.

\subsection{Odd $q$} \label{ss: odd}
Let $r$ be a positive integer and $\mathcal{A},\mathcal{B}\subseteq \9F_q^n$ so that $q=2r+1$ is an odd prime power. Then, we pick
\[
\mathcal{A}_1\eqdef\mathcal{A}_2\eqdef \mathcal{A}^{(r)},\qquad \mathcal{B}_1\eqdef\mathcal{B}_2\eqdef \mathcal{B}^{(r)}.
\]
We note that in this case the construction boils down to the adjacency matrices and the hulls of the $r$-th power codes, so we need to have a small enough $k$, such that $\bcode^{(r)}$ is not the whole space.
It is easy to see that, if $\6B=\6S\6A\6D\6P$, then $\6B_1=\6B_2=\6S_1 \6A_1\6D^r\6P=\6S_2 \6A_2\6D^r\6P$.
As shown in \cite{shur}, for randomly sampled codes $\mathcal{A},\mathcal{B}$ the dimension of $\7H(\mathcal{A}^{(r)})$ and  $\7H(\mathcal{B}^{(r)})$ is 0 with high probability.  Since $q-1=2r$ which clearly divides $  2r=i+j$, Proposition~\ref{prop: generalapproach_LEP} thus implies that Algorithm~\ref{alg:diag} provides a Probabilistic Polynomial-Time (PPT) algorithm for the average LEP with high probability of success, for $q=2r+1$ and whenever $\binom{k+r-1}{r}<n$, which for fixed $q$ asymptotically translates to
\[k< \sqrt[r]{r!\cdot  n}\cdot (1+o(1)).\]

\subsection{Frobenius automorphisms for non-prime fields} \label{ss: frobenius}
By letting the field size $q$ increase, the higher powers of codes used in the previous subsection quickly reduce the applicability range of the LEP solver.
The rate range loss can be mitigated by means of the Frobenius automorphisms.
Given $q=p^m$, where $p$ is a prime, and a non-negative integer $i$, we consider the image of a $[n,k]_q$ linear code $\acode$ through the $i$-th Frobenius automorphism $\phi_i$, where $\phi_i(x)=x^{p^i}$, (applied component-wise to each codeword):
\[
\acode^{[p^i]} = \{ \phi_i(\6x) \mid \6x \in \acode\}\subseteq \mathbb F_q^n,
\]
which can be readily verified being a $ [n,k]_{p^m}$ linear code.
\footnote{The code $\acode^{[p^i]}$ must not be confused with the power $\acode^{(p^i)}$.}

\begin{theorem} \label{th: frobenius_mon} Let $\acode, \bcode$ be $[n,k]_{q}$ linearly equivalent codes with generator matrices $\6A$ and $\6B$ respectively such that
 $\6B = \6S \6A \6D \6P$ for $\bS \in \mathrm{GL}_k(\mathbb{F}_q)$, $\bD = \mathrm{diag}(\mathbf d)$ with $\mathbf d \in (\F_q^\star)^n$, and $\bP \in S_n$.
Then, $\mathcal{B}^{[p^j]}, \mathcal{A}^{[p^j]}$ are linearly equivalent and their generator matrices $\6A^{[p^j]},\6B^{[p^j]}$ are such that $\6B^{[p^j]} = \6S^{[p^j]} \6A^{[p^j]} \6D^{p^j} \6P. $
\end{theorem}

Thus, using the Frobenius automorphism we have the following construction: let $\mathcal{A},\mathcal{B}\subseteq \9F_q^n$, with $q=p^m$. 
Then, we pick
\[
\mathcal{A}_1\eqdef\mathcal{A}_2\eqdef \mathcal{A}^{(p-1)}\star \dots\star(\mathcal{A}^{(p-1)})^{[p^{m-1}]}
\]\[
 \mathcal{B}_1\eqdef\mathcal{B}_2\eqdef \mathcal{B}^{(p-1)}\star \dots\star(\mathcal{B}^{(p-1)})^{[p^{m-1}]}.
\]
We have:
\begin{align*}
\dim_{\9F_q} (\mathcal{A}_1) &\le \prod_{i=0}^{m-1} \dim_{\9F_q }((\mathcal{A}^{(p-1)})^{[p^{i}]})\\ & =\prod_{i=0}^{m-1} \dim_{\9F_q }(\mathcal{A}^{(p-1)}) \le \binom{k+p-2}{p-1}^m.
\end{align*}
Once again, $\mathcal{A}_1\cap\mathcal{A}_2^\perp=\7H(\mathcal{A}_1), \mathcal{B}_1\cap\mathcal{B}_2^\perp=\7H(\mathcal{B}_1)$ and  we expect the dimension of $\7H(\mathcal{A}_1)$ and $\7H(\mathcal{B}_1)$ to be 0 with high probability.
Note that
$$\sum_{i=0}^{m-1} (p-1)p^i=(p-1)\frac{p^m-1}{p-1}=q-1.$$
Since $(q-1)\mid 2(q-1)=i+j$, Proposition~\ref{prop: generalapproach_LEP} thus implies that Algorithm~\ref{alg:diag} provides a PPT algorithm for the average LEP with high probability of success, for $q=p^m$ and whenever $\binom{k+p-2}{p-1}^m<n$, which for fixed $q$ asymptotically becomes
\[k < \sqrt[(p-1)m]{(p-1)!^m \cdot n}\cdot (1+o(1)).\]
This is the first case where $(q-1) \mid (i+j)$, but $q-1\ne i+j$. It is not a coincidence that, differently from the previous constructions, this choice also allows one to attack instances defined over fields of characteristic 2.

On the other hand, for odd field size, the factor 2 in $i+j=2(q-1)$ can be removed,  extending the range of weak parameters.
Indeed, let $q=p^m$, $p=2r+1$ and
\[
\mathcal{A}_1\eqdef\mathcal{A}_2\eqdef \mathcal{A}^{(r)}\star \dots\star(\mathcal{A}^{(r)})^{[p^{m-1}]},
\]
\[
\mathcal{B}_1\eqdef\mathcal{B}_2\eqdef \mathcal{B}^{(r)}\star \dots\star(\mathcal{B}^{(r)})^{[p^{m-1}]}.
\]
Similarly as before, the dimension of $\mathcal{A}_1=\mathcal{A}_2$ (and analogously of $\mathcal{B}_1=\mathcal{B}_2$) can be upper bounded as
\begin{align*} 
\dim_{\9F_q} (\mathcal{A}_1)& \le \prod_{i=0}^{m-1} \dim_{\9F_q }((\mathcal{A}^{(r)})^{[p^{i}]}) \le \binom{k+r-1}{r}^m.
\end{align*}
As before, we have that Algorithm~\ref{alg:diag} provides a PPT algorithm for the average LEP with high probability of success, for $q=(2r+1)^m$ and whenever $\binom{k+r-1}{r}^m<n$, which for fixed $q$ asymptotically becomes
\[k < \sqrt[rm]{r!^m \cdot n} \cdot (1+o(1)).\]

\subsection{Hermitian hulls and further} \label{ss: hermitian}
The analysis shown before by means of the Frobenius automorphism can be further improved for fields that are even-degree extensions of   prime fields, as we show next.

Let $q=p^{2\ell}$ and
\[
\mathcal{A}_1\eqdef \mathcal{A}^{(p-1)}\star \dots\star(\mathcal{A}^{(p-1)})^{[p^{\ell-1}]},\quad\mathcal{A}_2\eqdef \mathcal{A}_1^{[p^\ell]}
\]
\[\mathcal{B}_1\eqdef \mathcal{B}^{(p-1)}\star \dots\star(\mathcal{B}^{(p-1)})^{[p^{\ell-1}]},\quad \mathcal{B}_2\eqdef \mathcal{B}_1^{[p^\ell]}.
\]
With the same argument as before, the dimension of code \newline$\mathcal{A}_1$ (and analogously of $\mathcal{A}_2,\mathcal{B}_1,\mathcal{B}_2$) can be upper bounded as
\begin{align*}
\dim_{\9F_q} (\mathcal{A}_1 )&\le \prod_{i=0}^{\ell-1} \dim_{\9F_q }((\mathcal{A}^{(p-1)})^{[p^{i}]})
\le \binom{k+p-2}{p-1}^\ell.
\end{align*}
Differently from the previous constructions, this time we set $\mathcal{A}_1\ne\mathcal{A}_2, \mathcal{B}_1\ne\mathcal{B}_2$.
However, the target intersections can be reinterpreted as Hermitian hulls: $\mathcal{A}_1\cap\mathcal{A}_2^\perp=\7H_H(\mathcal{A}_1), \mathcal{B}_1\cap\mathcal{B}_2^\perp=\7H_H(\mathcal{B}_1)$.
Such Hermitian hulls are expected to be trivial with high probability as well. Let us define
$$t=\sum_{i=0}^{\ell-1} (p-1)p^i=(p-1)\frac{p^\ell-1}{p-1}.$$
Since
$$q-1=(p-1)\frac{p^{2\ell}-1}{p-1}= \sum_{i=0}^{2\ell-1} (p-1)p^i =(t+p^\ell t),$$
which clearly divides $(t+p^\ell t)=i+ j,$
and since $\dim_{\9F_q}(\mathcal{A}_1)=\dim_{\9F_q}(\mathcal{A}_2)$, Proposition~\ref{prop: generalapproach_LEP} thus implies that Algorithm~\ref{alg:diag} provides a PPT algorithm for the average LEP with high probability of success, for $q=p^{2\ell}$ and whenever $\binom{k+p-2}{p-1}^\ell<n$.
For $q$ fixed this asymptotically becomes
\[k < \sqrt[(p-1)\ell]{(p-1)!^\ell \cdot n} \cdot (1+o(1)).\]

\subsection{Odd degree extensions} \label{ss: odd_deg}

Finally, we show an additional improvement, applicable to odd degree extensions of odd characteristic fields, which cannot be interpreted as a hull (neither Euclidean nor Hermitian). Our framework is here applied in its most general form.

Let $q=p^{2\ell+1}$, $p=2r+1$ and
\[
\mathcal{A}_1\eqdef \mathcal{A}^{(p-1)}\star \dots\star(\mathcal{A}^{(p-1)})^{[p^{\ell-1}]}\star (\mathcal{A}^{(r)})^{[p^\ell]},
\]\[
\mathcal{A}_2\eqdef (\mathcal{A}^{(r)})^{[p^\ell]}\star(\mathcal{A}^{(p-1)})^{[p^{\ell+1}]}\star \dots\star(\mathcal{A}^{(p-1)})^{[p^{2\ell}]},\]
\[
\mathcal{B}_1\eqdef \mathcal{B}^{(p-1)}\star \dots\star(\mathcal{B}^{(p-1)})^{[p^{\ell-1}]}\star (\mathcal{B}^{(r)})^{[p^\ell]},
\]\[
\mathcal{B}_2\eqdef (\mathcal{B}^{(r)})^{[p^\ell]}\star(\mathcal{B}^{(p-1)})^{[p^{\ell+1}]}\star \dots\star(\mathcal{B}^{(p-1)})^{[p^{2\ell}]}.
\]
The dimension of $\mathcal{A}_1$ (and analogously of $\mathcal{A}_2,\mathcal{B}_1,\mathcal{B}_2$) can be upper bounded in the following way
\begin{align*}
\dim_{\9F_q}( \mathcal{A}_1) &\le  \dim_{\9F_q }((\mathcal{A}^{(r)})^{[p^{\ell}]})\cdot\prod_{i=0}^{\ell-1} \dim_{\9F_q }((\mathcal{A}^{(p-1)})^{[p^{i}]})\\
&\le \binom{k+p-2}{p-1}^\ell\cdot \binom{k+r-1}{r}.
\end{align*}
If $\6B=\6S\6A\6D\6P$, then $\6B_1=\6S_1 \6A_1\6D^{i}\6P$, $\6B_2=\6S_2 \6A_2\6D^{j}\6P$, where $i=(p-1)+(p-1)p+\dots+(p-1)p^{\ell-1}+rp^\ell$ and $j=rp^\ell+(p-1)p^{\ell+1}+\dots+(p-1)p^{2\ell}$. As usual, we can check that $q-1=(i+j)\mid (i+j)$ and that $\dim_{\9F_q}(\mathcal{A}_1)=\dim_{\9F_q}(\mathcal{A}_2)$.
Experiments reported in Table \ref{tab: LEP_prob} indicate that we get with good probability $\mathcal{A}_1\cap\mathcal{A}_2^\bot=\mathcal{B}_1\cap\mathcal{B}_2^\bot=\{\mathbf{0}\}$.
Similar to \cite{shur}, such probabilities could be further improved by applying a refinement as in the Support Splitting Algorithm \cite{SSA}. 
Thus, Proposition~\ref{prop: generalapproach_LEP} implies that Algorithm~\ref{alg:diag} provides a PPT algorithm for the average LEP with high probability of success, for $q=(2r+1)^{2\ell+1}$ and whenever $\binom{k+p-2}{p-1}^\ell\cdot \binom{k+r-1}{r}=\binom{k+2r-1}{2r}^\ell \cdot \binom{k+r-1}{r}<n$, which asymptotically relaxes the constraint on $k$ to
\[k < \sqrt[2r \ell +r]{(2r)!^\ell r! \cdot n}\cdot  (1+o(1)).\]
Table~\ref{tab: LEP_summary} summarizes the range of parameters that can be attacked within our framework, depending on the field size~$q$.
\begin{table}[t]
\caption{Asymptotic upper bound for $n \to \infty$ and $q$ fixed on the dimension that can be attacked by Algorithm~\ref{alg:diag}}
\label{tab: LEP_summary}
\centering
\begin{tabular}{|c|c|}
\hline
Field size form & Asymptotic upper bound on the dimension\\
\hline 
$q=2r+1$ & $k< \sqrt[r]{r!\cdot  n}\cdot (1+o(1))$ \\
$q=2^m$ & $k < \sqrt[m]{n} \cdot (1+o(1))$ \\
$q=(2r+1)^m$ & $k < \sqrt[rm]{r!^m \cdot n} \cdot (1+o(1))$ \\
$q=2^{2\ell}$ & $k < \sqrt[\ell]{n} \cdot (1+o(1))$ \\
$q=(2r+1)^{2\ell+1}$ & $k < \sqrt[2r \ell+r]{(2r)!^\ell r! \cdot n} \cdot (1+o(1))$ \\ \hline
\end{tabular}
\end{table}
Table~\ref{tab: LEP_prob} provides information about the failure rate of our probabilistic algorithm on two random codes $\mathcal{A},\mathcal{B}$ for several parameters. In particular, $T$ is the event that both $\mathcal{A}_1\cap \mathcal{A}_2^\bot$ and $\mathcal{B}_1\cap \mathcal{B}_2^\bot$ are trivial and $M$ denotes the event $\mathsf{Count}(x_{1,1}, \ldots,x_{n,n})\neq \mathsf{Count}(y_{1,1}, \ldots, y_{n,n})$. The experimental conditional probability $\9P(M \mid T)$ is taken over $10^6$ samples. This can be compared with the theoretical estimate given by Eq.~\eqref{eq: prob_countdiag}. In this case, the value of $q$ is not necessarily the size of the field where the codes are defined, but rather the size of the subfield where the diagonal elements $x_{i,i},y_{i,i}$ live, which depends on the definition of $\mathcal{A}_i,\mathcal{B}_i$. For parameters $[300,6]_8$, the diagonal elements always belong to $\9F_2$, for $[100, 12]_9$ to $\9F_3$  and for $[100,8]_{16}$ to $\9F_4$. For all parameters, the discrepancy amounts to less than one order of magnitude.
\renewcommand{\tabcolsep}{5pt}
\begin{table}[t]
\caption{Test failure rate on random codes}
\label{tab: LEP_prob}
\centering
\begin{tabular}{|c|c|c|c|c|}
\hline
Parameters & Construction & $\9P(T)$ & experimental & estimate \\
$[n,k]_q$&&&$\9P(M\mid T)$ & $\9P(M)$ \\ \hline
$[100,10]_5$ &Subsection~\ref{ss: odd}& 0.630 & $1.84\cdot 10^{-4}$ & $3.54\cdot 10^{-5}$ \\
 $[300,6]_8$ & Subsection~\ref{ss: frobenius}& 0.175 & 0.0646 & 0.0326 \\
$[100, 12]_9$ &Subsection~\ref{ss: hermitian}& 0.518 & 0.0125 & $4.135\cdot 10^{-3}$\\
$[100,8]_{16}$ &Subsection~\ref{ss: hermitian}& 0.619 & $1.40\cdot 10^{-3}$ & $3.59\cdot 10^{-4}$ \\
\hline
\end{tabular}
\end{table}

\section{Partial Closure}
To conclude, we introduce a family of intermediate closure operations, referred to as partial closures, and show that they yield a reduction of a broader class of linear equivalence problems to permutation equivalence.

Note that the $\ell$-th power code transforms two linearly equivalent codes into two linearly equivalent codes where the entries of the monomial matrix are in the group $U$ of $\ell$-th powers in $\mathbb{F}_q^\star$.

\begin{definition}
    The (search version of the) \textit{code equivalence problem with entries in the subgroup} $U \subseteq \F_q^\star$ (written LEP($U$)) asks, given two codes $\acode, \bcode$ with generator matrices $\bA, \bB \in \mathbb F_q^{k\times n}$, whether there exist an invertible matrix $\bS \in \mathrm{GL}_k(\mathbb{F}_q)$, a diagonal matrix $\bD = \mathrm{diag}(\mathbf d)$ with $\mathbf d \in U^n$, and a permutation matrix $\bP \in S_n$ such that $\bB = \bS \bA \bD \bP$, and to find such a triple.
\end{definition}

LEP, PEP and the Signed Permutation Equivalence Problem fall naturally in LEP($U$) with $U = \F_q^\star, U = \{1\}$ and $U =\{\pm 1\}$.

A natural question to ask is what the complexity of this problem is. For this, we give a reduction of LEP($U$) to permutation equivalence in the following.

\begin{definition}
\label{def:partial_closure}
    Let $\alpha\in \F_q^\star$ be a primitive element, and let $\acode$ be a $[n,k]_q$ linear code generated by $\bA\in \mathbb{F}_q^{k \times n}$.
    Let $r \mid (q-1)$ and define 
    \[
        \adjlpha\eqdef(1,\alpha^{(q-1)/r}, \alpha^{2(q-1)/r}, \ldots, \alpha^{(r-1)(q-1)/r}).
    \]
    The $r$-th \emph{partial closure} is defined as the span of $\adjlpha \otimes \bA$ and denoted by $\adjlpha \otimes \acode $.
\end{definition}

\begin{remark}
    The special case $r = q - 1$ is the usual closure and the case $r = 1$ does not change the code at all.
\end{remark}

Let $U = \langle \alpha^{(q-1)/r} \rangle$ be the subgroup of $\mathbb{F}_q^\star$ generated by $\alpha^{(q-1)/r}$.
Note that for $\adjlpha$ as above, we still have $\adjlpha \adjlpha^\top = \sum_{\beta \in U} \beta^2=0$ as long as $|U|=r>2$ by orthogonality of characters since there exists an element $\beta$ such that $\beta^2 \ne 1$ and since the map $x \to \beta x $ is a permutation of $U$.

Now we can apply the partial closure.
Note that while we still end up with a self-orthogonal code, the length is now smaller and we get a better reduction to PEP than by just taking the closure.
\begin{theorem}
    Let $\acode, \bcode$ be two $[n,k]_q$ linear codes and $\adjlpha$ as in Definition \ref{def:partial_closure}.
    Then $\acode$ is equivalent to $ \bcode$ with entries in the subgroup $U$  generated by $\alpha^{\frac{q-1}{r}}$ if and only if 
    $\adjlpha \otimes \acode$ is permutation equivalent to $\adjlpha \otimes \bcode$.
\end{theorem}

The proof for the reduction is essentially the same as for Theorem \ref{thm:reduction_LEP_PEP} in \cite[Theorem 1]{hull} and thus omitted. 

\section*{Acknowledgements}

Michele Battagliola and Paolo Santini are supported by the Italian Ministry of University and Research (MUR) under the PRIN 2022 program with projects ``Mathematical Primitives for Post Quantum Digital Signatures'' \\ (CUP I53D23006580001), ``Post quantum Identification and eNcryption primiTives: dEsign and Realization (POINTER)'' (CUP I53D23003670006) and  the Italian Fund for Applied Science (FISA 2022), project “Quantum-safe cryptographic tools for the protection of national data and information technology assets” (QSAFEIT) - No. FISA 2022-00618 (CUP I33C24000520001).
Michael Schaller is supported by the Swiss National Science Foundation under SNSF grant number 212865.
Violetta Weger  is supported by the Technical University of Munich – Institute for Advanced Study, funded by the German Excellence Initiative. 

This paper is the result of merging two independently conducted works, developed in parallel and later unified into this manuscript. Anna-Lena Horlemann, Abhinaba Mazumder, Michael Schaller and Violetta Weger would like to thank the organizers and sponsors of the Virginia Tech-Swiss Summer School 2024, where parts of this work were initiated.

\printbibliography

@inproceedings{biasse2020less,
  title={{LESS is more: code-based signatures without syndromes}},
  author={Biasse, Jean-Fran{\c{c}}ois and Micheli, Giacomo and Persichetti, Edoardo and Santini, Paolo},
  booktitle={International Conference on Cryptology in Africa},
  pages={45--65},
  year={2020},
  organization={Springer}
}

@article{barenghi2022advanced,
  title={Advanced signature functionalities from the code equivalence problem},
  author={Barenghi, Alessandro and Biasse, Jean-Fran{\c{c}}ois and Ngo, Tran and Persichetti, Edoardo and Santini, Paolo},
  journal={International Journal of Computer Mathematics: Computer Systems Theory},
  volume={7},
  number={2},
  pages={112--128},
  year={2022},
  publisher={Taylor \& Francis}
}

@article{budroni2025two,
  title={Two Is All It Takes: Asymptotic and Concrete Improvements for Solving Code Equivalence},
  author={Budroni, Alessandro and Esser, Andre and Franch, Ermes and Natale, Andrea},
  journal={Cryptology ePrint Archive},
  year={2025}
}

@article{baldi2025speck,
  title={{SPECK: Signatures from Permutation Equivalence of Codes and Kernels}},
  author={Baldi, Marco and Battagliola, Michele and El Mechri, Rahmi and Santini, Paolo and Schiavoni, Riccardo and De Zuane, Davide},
  journal={Cryptology ePrint Archive},
  year={2025}
}

@article{battagliola2025vole,
  title={{VOLE-in-the-Head Signatures Based on the Linear Code Equivalence Problem}},
  author={Battagliola, Michele and Mattiuz, Laura and Meneghetti, Alessio},
  journal={Cryptology ePrint Archive},
  year={2025}
}

@inproceedings{hanzlik2025tanuki,
  title={{Tanuki: New Frameworks for (Concurrently Secure) Blind Signatures from Post-Quantum Group Actions}},
  author={Hanzlik, Lucjan and Lai, Yi-Fu and Mula, Marzio and Paracucchi, Eugenio and Slamanig, Daniel and Tang, Gang},
  booktitle={International Conference on the Theory and Application of Cryptology and Information Security},
  pages={35--69},
  year={2025},
  organization={Springer}
}

@inproceedings{sendrier2013hardness,
  title={The hardness of code equivalence over $\mathbb{F}_q$ and its application to code-based cryptography},
  author={Sendrier, Nicolas and Simos, Dimitris E},
  booktitle={International Workshop on Post-Quantum Cryptography},
  pages={203--216},
  year={2013},
  organization={Springer}
}

@inproceedings{beullens2020not,
  title={Not enough LESS: an improved algorithm for solving code equivalence problems over {$\mathbb{F}_q$}},
  author={Beullens, Ward},
  booktitle={International Conference on Selected Areas in Cryptography},
  pages={387--403},
  year={2020},
  organization={Springer}
}

@article{chou2025linear,
  title={On linear equivalence, canonical forms, and digital signatures},
  author={Chou, Tung and Persichetti, Edoardo and Santini, Paolo},
  journal={Designs, Codes and Cryptography},
  pages={1--43},
  year={2025},
  publisher={Springer}
}

@inproceedings{nowakowski2025improved,
  title={An improved algorithm for code equivalence},
  author={Nowakowski, Julian},
  booktitle={International Conference on Post-Quantum Cryptography},
  pages={71--103},
  year={2025},
  organization={Springer}
}

@article{SSA,
  title={Finding the permutation between equivalent linear codes: The support splitting algorithm},
  author={Sendrier, Nicolas},
  journal={IEEE Transactions on Information Theory},
  volume={46},
  number={4},
  pages={1193--1203},
  year={2000},
  publisher={IEEE}
}

@phdthesis{saeed2017algebraic,
  title={Algebraic approach for code equivalence},
  author={Saeed, Mohamed Ahmed},
  year={2017},
  school={Normandie Universit{\'e}; University of Khartoum}
}

@inproceedings{magali,
  title={Permutation code equivalence is not harder than graph isomorphism when hulls are trivial},
  author={Bardet, Magali and Otmani, Ayoub and Saeed-Taha, Mohamed},
  booktitle={2019 IEEE International Symposium on Information Theory (ISIT)},
  pages={2464--2468},
  year={2019},
  organization={IEEE}
}

@inproceedings{hull,
  title={How easy is code equivalence over $\mathbb{F}_q$?},
  author={Sendrier, Nicolas and Simos, Dimitrios E},
  booktitle={International Workshop on Coding and Cryptography-WCC 2013},
  year={2013},
  url={https://inria.hal.science/hal-00790861/document}
}

@book{Huffman_Pless_2003, 
place={Cambridge}, 
title={Fundamentals of Error-Correcting Codes}, publisher={Cambridge University Press},
author={Huffman, W. Cary and Pless, Vera}, 
year={2003}
}

@misc{shur,
      author = {Michele Battagliola and Rocco Mora and Paolo Santini},
      title = {Using the {Schur} Product to Solve the Code Equivalence Problem},
      howpublished = {Cryptology {ePrint} Archive, Paper 2025/1017},
      year = {2025},
      url = {https://eprint.iacr.org/2025/1017}
}

@article{SendrierHull,
author = {Sendrier, Nicolas},
title = {On the Dimension of the Hull},
journal = {SIAM Journal on Discrete Mathematics},
volume = {10},
number = {2},
pages = {282-293},
year = {1997},
URL = {https://doi.org/10.1137/S0895480195294027}
}

@InProceedings{albrecht2025hollow,
author="Albrecht, Martin R.
and Ben{\v{c}}ina, Benjamin
and Lai, Russell W. F.",
editor="Fehr, Serge
and Fouque, Pierre-Alain",
title="{Hollow LWE: A New Spin}",
booktitle="Advances in Cryptology -- EUROCRYPT 2025",
year="2025",
publisher="Springer Nature Switzerland",
address="Cham",
pages="363--392",
isbn="978-3-031-91101-9"
}

@article{richmond2009counting,
  title={Counting {Abelian} Squares},
  author={Richmond, LB and Shallit, Jeffrey},
  journal={the electronic journal of combinatorics},
  volume={16},
  number={R72},
  pages={1},
  year={2009},
  publisher={Citeseer}
}

\end{document}